\documentclass[conference]{IEEEtran}%
\usepackage{amsfonts}
\usepackage{amsmath}
\usepackage{amssymb}
\usepackage{graphicx}%
\providecommand{\U}[1]{\protect\rule{.1in}{.1in}}
\newtheorem{theorem}{Theorem}

\newtheorem{corollary}{Corollary}

\newtheorem{lemma}{Lemma}

\newtheorem{remark}{Remark}

\begin{document}

\title{Beyond the Bethe Free Energy of LDPC Codes via Polymer Expansions}

\author{\authorblockN{Nicolas Macris and Marc Vuffray}
\authorblockA{LTHC-IC-EPFL\\
Lausanne, Switzerland\\
nicolas.macris@epfl.ch, marc.vuffray@epfl.ch}}
%

\maketitle
%

\begin{abstract}%
The loop series provides a formal way to write down corrections to the Bethe entropy (and/or free energy)
of graphical models. We provide methods to rigorously control such expansions for low-density parity-check codes used over a 
highly noisy binary symmetric channel. We prove that
in the asymptotic limit of large size, with high probability, the Bethe expression gives an exact formula for the entropy (per bit) of the input word
conditioned on the output of the channel. Our methods also apply to more general models.

\end{abstract}%

\section{Introduction\label{sec:introduction}}

Often one needs to compute the free energy and/or entropy of a graphical model. 
The Bethe approximation and the related Belief Propagation (BP) equations may sometimes 
offer a good starting point. However it is seldom a controlled approximation and even worse it is usually not clear if it
yields upper or lower bounds, or even if there is any such relationship. There are not many results that precisely pinpoint the relation between the Bethe 
and true free energies or entropies. A general result of Vontobel \cite{Vontobel} relates the Bethe free energy to an average of the true free energy over all graph covers.
For Ising-like graphical models with attractive pair interactions, 
Wainwright \cite{Wainwright} has shown that, under additional special conditions, the Bethe free energy is a bound to the true free energy. This work 
uses the same loop series used here. 
It is well known that the Bethe free energy is exact on 
trees, and it is natural to investigate its possible exactness on random Erdoes-R\'enyi type graphs which are known to be locally tree-like. But we already know 
of systems, such as random constraint satisfaction models (e.g, $K$-SAT or $Q$-coloring) or spin glasses, where the true free energy {\it is not} given by the Bethe formula - 
even when averaged over the graph ensemble. The local tree-like nature of the graph is not sufficient when long ranged correlations are present \cite{mezard-montanari}.

For graphical models that describe 
communication with low density (parity-check and generator-matrix) codes over binary-symmetric memoryless channels the situation 
is favorable. Indeed we have plenty of evidence 
that the replica-symmetric solution\footnote{Replica-symmetric formulas are averaged forms of the Bethe 
formulas, where the average is over the channel output realizations and code ensemble.} is exact. See \cite{Montanari}, \cite{Macris-1}, \cite{Kudekar-Macris} 
for bounds and \cite{Rudiger}, \cite{Korada-Macris-Kudekar} for results on the binary erasure channel. In \cite{kudekar-Macris-SIAMpaper} it is proven 
that correlations between 
pairs of distant (with respect to Tanner graph distance) bits decay exponentially fast for LDGM codes in the regime of large noise, and LDPC codes in 
the regime of small noise. This also allowed to conclude that 
the replica symmetric formulas are exact in these regimes. 

A few years ago Chertkov and Chernyak \cite{Chertkov} developed a loop series representation for the free energy of graphical models. The virtue of 
this representation is that it isolates 
the Bethe contribution, and represents the {\it remainder} by a series of terms involving only BP messages associated to {\it generalized loops} of the graph. 
It is tempting to use this representation as a tool to compare the true and Bethe free energies. 


In this contribution we consider regular LDPC$(l,r)$ codes used over a {\it highly noisy} BSC. Consider 
the conditional entropy $\frac{1}{n}H(\underline X\vert \underline Y)$ of the input word $\underline X=(X_1\cdots X_n)$ given 
a channel output $\underline Y=(Y_1\cdots Y_n)$. 
We prove that in the large size limit, with high probability with respect to the code ensemble, the difference between the conditional entropy and the Bethe formula tends to zero.
The error term essentially comes from the probability that the graph is not locally tree-like. Our techniques also allow to organize 
the dominant correction terms into a {\it polymer expansion}\footnote{See \cite{Brydges} for a pedagogical introduction to polymer expansions.} involving generalized loops of 
size less than $\lambda_0 n$ ($0<\lambda_0<1$ a constant). 
As we will show, expander arguments imply that this polymer expansion 
converges uniformly in $n$. When the terms of the polymer expansion are added to the Bethe expression, with high probability,
the difference with the conditional entropy becomes $O(e^{-n \epsilon})$ for some $\epsilon>0$. 

Our results also apply to more general models. Namely the channel could have asymmetric flip probability. In fact the whole technique 
and results apply to spin-glass models on $(l,r)$ Tanner graphs with $l$ odd and $l<r$, with {\it small 
magnetic fields}, and {\it any temperature}.The 
limitation to $l<r$ is not just technical. Indeed $l>r$ would correspond
to a kind of XORSAT constraint satisfaction problem, and for the usual XORSAT problem we know that the replica symmetric solutions 
are not generally exact at low temperatures.

The case $l=2$ (cycle codes) has its own special features and has been discussed in \cite{Macris Vuffray}.

\section{Preliminaries}\label{sec2}

We begin with a few definitions and notations. Fix two integers 
$l<r$. Consider two vertex 
sets: $V$ a set of $n$ 
{\it variable nodes} and $C$  a set of $m=n\frac{l}{r}$ {\it check nodes}. 
We think of $n$ large and $l,r$ fixed. 
We consider bipartite $(l,r)$ regular graphs - call them $\Gamma$ - connecting $V$ and $C$. The set of edges is $E$. More precisely, vertices of $V$ have degree $l$, vertices of $C$ have degree $r$,
and there are no double edges. The set of all such graphs is denoted $\mathcal{B}(l,r,n)$. Note that $\Gamma$ is the Tanner graph of a LDPC code with design rate $1-l/r$.
 When we say that $\Gamma$ is random we mean that we 
draw it uniformly randomly from the set $\mathcal{B}\left(l,r,n\right)$. The corresponding expectation is $\mathbb{E}_\Gamma$. 

Letters $i,j$ will always denote nodes 
in $V$ and letters $a,b$ nodes in $C$.
We reserve the notations 
$\partial i$ (resp. $\partial a$) for the sets of neighbors of $i$ (resp. $a$) in $\Gamma$.

We will say that $\Gamma$ is a $(\lambda,\kappa)$ expander if for every subset $\mathcal{V}\subset V$ such that $\left\vert
\mathcal{V}\right\vert <\lambda n$ we have $\left\vert \partial\mathcal{V}
\right\vert \geq\kappa l\left\vert \mathcal{V}\right\vert $. Here $\partial\mathcal{V}$ is the number of check nodes that are connected to $\mathcal{V}$.
Take a random $\Gamma$. We can always find $\lambda >0$ such that with probability 
$1- O(n^{-(l(1-\kappa) -1)})$, $\Gamma$ is a $(\lambda,\kappa)$ expander with $\kappa< 1-\frac{1}{l}$. It is sufficient  to take  
$0<\lambda <\lambda_{0}$ where $\lambda_{0}$ is 
the positive solution of the equation\footnote{See e.g \cite{Rudiger} where the standard LDPC$(l,r,n)$ ensemble is considered. It is easily argued that the same result 
applies to $\mathcal{B}\left(l,r,n\right)$.}
\begin{equation}
 \frac{l-1}{l} h_2(\lambda_0) - \frac{l}{r} h_2(\lambda_0\kappa r) -\lambda_0\kappa r h_2(\frac{1}{\kappa r}) = 0\,.
\end{equation}
As will be seen later we need to take $\kappa\in ]1-\frac{2(r-1)}{lr}, 1-\frac{1}{l}[$ (which is always possible for $r>2$). In the rest of the paper
$\kappa$ is always a constant in this interval, and $0<\lambda<\lambda_0$.
For concreteness, one can take the case
$\left(  l,r\right)  =\left(  3,6\right)$, fix $\kappa =1/2$ and $\lambda_0 = 5\times 10^{-4}$.

Assume that we transmit (with uniform prior) code words from an LDPC code with Tanner graph $\Gamma$ over a BSC with flip probability $p$. 
We assume without loss of generality
that the all zero codeword is transmitted. 
Then the posterior probability that $\underline x=(x_i)_{i =1}^n\in \{0,1\}^n$ is the transmitted word given that 
$\underline y =(y_i)_{i=1}^n\in \{0,1\}^n$ is received, reads 
\begin{equation}
p_{\underline{X}|\underline{Y}}\left(  \underline{x}|\underline{y}\right)
= \frac{1}{Z}\prod_{a\in C}\mathbb{I}\left(  \oplus_{i\in\partial a}x_{i}=0\right)
\prod_{i\in V}\exp(( -1)^{x_{i}} h_i  )\,. \label{eq:ldpc_dist}
\end{equation}
The graph $\Gamma$ enters in this formula through the parity check constraints. 
In this formula $h_i = (-1)^{y_i}\frac{1}{2}\ln\frac{1-p}{p}$ and $Z$ is the normalizing factor
\begin{equation}
 Z = \sum_{\underline x\in\{0,1\}^n}\prod_{a\in C}\mathbb{I}\left(  \oplus_{i\in\partial a}x_{i}=0\right)
\prod_{i\in V}\exp((  -1)^{x_{i}} h_i) .
\end{equation}
We set 
\begin{equation}
 h=\frac{1}{2}\ln\frac{1-p}{p}
\end{equation}
It is good to keep in mind that the high noise regime considered in this paper corresponds to small $h$ ($p$ close to $1/2$) and that $\vert h_i\vert =h$.

It is equivalent to describe the channel outputs in terms of $\underline y$ or in terms of the half-log-likelihood variables $\underline h=(h_i)_{i=1}^n$.
Note that $h_i$ have the probability distribution $c(h_i) = (1-p)\delta(h_i - \ln\frac{1-p}{p}) + p\delta(h_i - \ln\frac{p}{1-p})$. The expectation with respect to this 
distribution is called $\mathbb{E}_{\underline h}$.
We are interested in the conditional entropy $H(\underline X\vert \underline Y)$ of the input word given the output word. We have 
(see e.g, \cite{mezard-montanari})
\begin{equation}\label{entro-free}
\mathfrak{h}_n\equiv\frac{1}{n}H\left(  \underline{X}|\underline{Y}\right)  = \frac{1}{n}\mathbb{E}
_{\underline{h}}\left[\ln Z\right]  - \frac{1-2p}{2}\ln\frac{1-p}{p} .
\end{equation}
In \eqref{entro-free}, $n^{-1}\ln Z$ is the {\it free energy} of the Gibbs measure \eqref{eq:ldpc_dist}.

\section{The Bethe Approximation}

The Bethe free energy 
 involves a set of 
messages $\left\{  \eta_{i\to a},\widehat{\eta}_{a\to i}\right\}$ attached to the edges of $\Gamma$. The collection of all messages 
is denoted $(\underline\eta, \underline{\widehat{\eta}})$.
These satisfy the BP equations
\begin{align}
\begin{cases}
\eta_{i\rightarrow a}  &  =h_{i}+\sum_{ b\in \partial i\backslash
a}\widehat{\eta}_{b\rightarrow i}\nonumber\\
\widehat{\eta}_{a\rightarrow i}  &  =\tanh^{-1}\bigl(  \prod_{j\in \partial a\setminus i}
\tanh \eta_{j\rightarrow a}  \bigr)  . \label{eq:bp_equations}
\end{cases}
\end{align}
These equations always have a trivial solution $\tanh\eta_{i\to a} = \tanh\widehat{\eta}_{a\to i}$=1. We will consider only 
non-trivial solutions that are relevant for small $h$. For these solutions $\eta_{i\to a}$ and $\widehat{\eta}_{a\to i}$ take small values and we can show that
$\vert\eta_{i\to a}\vert\leq \vert h\vert + (l-1)\vert h\vert^{r-1} +O(\vert h\vert^r)$ and $\vert\widehat{\eta}_{a\to i}\vert \leq \vert h\vert^{r-1} + O(\vert h\vert^r)$. We call such solutions {\it high-noise-solutions}.

These solutions have a Bethe free energy 
\begin{equation}
f_{\mathrm{Bethe}}\left(  \underline\eta,\underline{\widehat{\eta}}\right)  =\frac{1}{n}\biggl(\sum_{a\in
C}F_{a}+\sum_{i\in V}F_{i}-\sum_{\left(  i,a\right)  \in
E}F_{ia}\biggr), \label{eq:fbethe}
\end{equation}
where
\begin{align}
\begin{cases}
F_{a}    =\ln\frac{1}{2}(  1+\prod_{i\in\partial a}\tanh  \eta_{i\rightarrow
a})  +\sum_{i\in\partial a}\ln2\cosh  \eta_{i\rightarrow
a}  ,\nonumber\\
F_{i}    =\ln2\cosh\left(  h_{i}+\sum_{a\in\partial i}\widehat{\eta
}_{a\rightarrow i}\right)  ,\nonumber\\
F_{ia}    =\ln2\cosh\left(  \eta_{i\rightarrow a}+\widehat{\eta
}_{a\rightarrow i}\right)  . \label{eq:fi_fa_fia}
\end{cases}
\end{align}

\begin{theorem}\label{theorem1}
Suppose $l$ is odd and $3\leq l\leq r$. There exists $h_0>0$ (small) independent of $n$, 
such that for $\vert h\vert\leq h_0$ and any high-noise-solution 
$\left(  \underline\eta,\underline{\widehat{\eta}}\right)$ of the BP equations,
\begin{equation}\label{first-result}
\mathbb{E}_\Gamma[\vert \frac{1}{n}\ln Z - f_{\mathrm{Bethe}}\left(  \underline\eta,\underline{\widehat{\eta}}\right)\vert] = O\bigl(\frac{1}{n^{l(1-\kappa)-1}}\bigr)\,.
\end{equation}
The $O(\cdot)$ is uniform in the channel output realizations $\underline h$. 
\end{theorem}

\begin{remark}
 By Markov's bound we obtain that the difference between the true and Bethe free energies tends to zero with high probability, in the 
$n\to +\infty$ limit.
\end{remark}

\begin{remark}
We can average equation \eqref{first-result} over the channel output and use \eqref{entro-free} to relate the true and Bethe entropies.
\end{remark}

\section{Loop Corrections to the Bethe Approximation}
 
We define a {\it generalized loop} $g$ as any subgraph 
contained in $\Gamma$ with no dangling edges (figure \ref{fig:polymers}). 
Note that a 
generalized loop is not necessarily connected.
We call
$d_i(g)$ (resp. $d_a(g)$) the {\it induced degree} of 
node $i$ (resp. $a$) in $g$. For a generalized loop we have
$d_i(g)\in\{2,\cdots,l\}$ and 
$d_a(g)\in\{2,\cdots,r\}$. 

\begin{figure}[ptb]
\centering
\includegraphics[
height=1.05in,
width=2.00in
]
{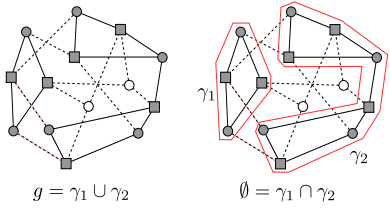}
\caption{Example of $\Gamma\in\mathcal{B}\left(  3,4,8\right)  $. The generalized loop $g$ has two
disjoint connected parts $\gamma_1$ and $\gamma_2$.}
\label{fig:polymers}
\end{figure}

For a finite size system, the {\it loop series} \cite{Chertkov} is an identity
valid for any solution of the BP equations. We have
\begin{equation}\label{cc-identity}
 \frac{1}{n}\ln Z - f_{\mathrm{Bethe}}\left(\underline\eta,\underline{\widehat{\eta}}\right)
= \frac{1}{n} \ln\biggl\{\sum_{g\subset\Gamma}K\left(
g\right)\biggr\}.
\end{equation}
The sum on the right hand side carries over all generalized loops included in $\Gamma$.
The $K(g)$ can be expressed entirely in terms of BP messages $\eta_{i\to a}$ and 
$\widehat\eta_{a\to i}$. The explicit formula is given in the appendix.
Remarkably $K(g)$ factorizes in a product of 
contributions associated to the connected parts of $g$. Each generalized loop can be decomposed in a unique way
as a union $g =\cup_{k}\gamma_k$ where $\gamma_k$ are {\it connected and disjoint generalized loops}. The $\gamma_k$'s are called 
{\it polymers}. We have $K(g)= \prod_k K(\gamma_k)$ and 
\begin{align}
\sum_{g\subset\Gamma}K\left(
g\right) =\sum_{M\geq 0} & \frac{1}{M!} 
\sum_{\gamma_{1},...,\gamma_{M}\subset \Gamma}\prod_{k=1}^{M}K\left(  \gamma_{k}\right)
\nonumber \\ &
\times
\prod_{k<k^{\prime}}\mathbb{I}\left(  \gamma_{k}\cap\gamma_{k^{\prime}
}=\emptyset\right)  . \label{eq: z polymer}
\end{align}
In the sum each $\gamma_{k}$ runs over all polymers contained in $\Gamma$. 
The factor $\frac{1}{M!}$ accounts for the fact that a
polymer configuration has to be counted only once. Finally the indicator function ensures that the polymers do not intersect.
Because of this constraint all sums in \eqref{eq: z polymer} are finite.

From a physical
point of view \eqref{eq: z polymer} is the partition function of polymers that can acquire any shape allowed by $\Gamma$, have {\it activity}\footnote{This is the name used by chemists to 
denote the probability weight assuming that the polymer would be isolated. Note that here $K(\gamma)$ can be negative and this analogy is at best formal.}
$K(\gamma)$, and interact via a two body hard-core repulsion. This analogy allows us to use methods from statistical mechanics to analyze the corrections to the Bethe 
free energy. 

We say that a {\it polymer is small} if $\vert\gamma\vert <\lambda n$ for some fixed $\lambda$ that we take in the interval $[0,\lambda_0]$. The contribution of small polymers to 
\eqref{eq: z polymer} is 
 \begin{align}
Z_{p}\left(  \underline\eta,\underline{\widehat{\eta}}\right)  =\sum_{M\geq0} & \frac{1}{M!}
\sum_{\gamma_{1},...,\gamma_{k}~\mathrm{s.t}~\left\vert \gamma
_{k}\right\vert <\lambda n}
\prod_{k=1}^{M}K\left(  \gamma_{k}\right)
\nonumber \\ &
\prod_{k<k^{\prime}}\mathbb{I}\left(  \gamma_{k}\cap\gamma_{k^\prime}=\emptyset\right).
\end{align}

\begin{theorem}\label{second-theorem}
Suppose $l$ is odd and $3\leq l\leq r$. take $\Gamma$ at random. There exist a small $h_0$ independent of $n$ such that for $\vert h\vert<h_0$, and any high-noise-solution 
$\left(  \underline\eta,\underline{\widehat{\eta}}\right)$ of the BP equations, with probability $1-\frac{1}{\epsilon}O(n^{-(l(1-\kappa)-1)})$, 
\begin{equation}\label{second-result}
\frac{1}{n}\ln Z = f_{\mathrm{Bethe}}\left(  \underline\eta,\underline{\widehat{\eta}}\right) + \frac{1}{n}\ln Z_p + O(e^{-\epsilon n})
\end{equation}
for $\epsilon>0$. Here $O(\cdot)$ is uniform $\underline h$.  
\end{theorem}

\begin{figure}
\centering
\includegraphics[
height=1.15in,
width=3.00in]
{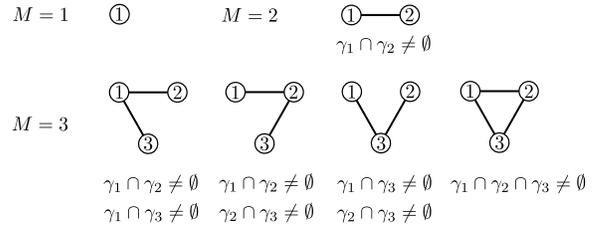}\caption{All the Mayer graphs for $M=1,2,3$.}
\label{fig mayer graph}
\end{figure}
 
The second term on the right hand side of \eqref{second-result} is the partition function of small polymers. One can compute in a systematic way
the leading corrections to the Bethe free energy by expanding the logarithm in powers of the activities $K(\gamma)$.
This yields the so-called {\it polymer (or Mayer) expansion},
\begin{align}
\frac{1}{n}\ln Z_{p}\left(  \vec{\eta}\right)   &  =\frac{1}{n}\sum_{M\geq 1}^{+\infty}\frac{1}{M!}
\sum_{\gamma_{1},...,\gamma_{M}~{\rm s.t}~\vert \gamma
_{k}\vert <\lambda n}\prod_{k=1}^{M}K\left(  \gamma_{k}\right)
\nonumber\\
&  \times\sum_{G\subset\mathcal{G}_{M}}\prod_{(k,k^{\prime})\in G}
(-\mathbb{I}\left(  \gamma_{k}\cap\gamma_{k^{\prime}}\neq\emptyset\right)  ).
\label{eq:polymer_expansion}
\end{align}
The third sum is over the set $\mathcal{G}_{M}$ of all connected {\it Mayer graphs} $G$ with M
vertices labeled by $\gamma_{1},...,\gamma_{M}$ (see figure \ref{fig mayer graph}). Note that in the expansion of the logarithm, the indicator
function forces the polymers to overlap. Therefore the summations contains an infinite number of terms and its convergence has to be controlled. 

\begin{lemma}\label{convergence-lemma}
Suppose $r>2$.
Fix $\zeta_0>1$ and replace $K(\gamma)$ by $\zeta K(\gamma)$ ($\zeta\in \mathbb{C}$) in the polymer expansion \eqref{eq:polymer_expansion} which then becomes 
a power series in the parameter $\vert \zeta\vert\leq \zeta_0$. Assume that $\Gamma$ is a $(\lambda, \kappa)$ expander with $\kappa\in]1-\frac{2(r-1)}{lr}, 1-\frac{1}{l}[$.
One can find $h_0>0$ such that for $\vert h\vert<h_0$ this power series is absolutely convergent uniformly in $n$ and $\underline h$.
\end{lemma}

\begin{remark}
This lemma holds for any $(l,r)$ with $r>2$.
\end{remark}

\begin{remark}
Our real interest is of course for $\zeta=1$, and the introduction of the parameter $\zeta$ above is just a convenient way to describe the nature of the polymer expansion.
The lemma implies that one can compute the limit $n\to+\infty$ of the polymer expansion {\it term by term} (for small polymers), and that this limit is analytic for 
$\vert \zeta\vert <\zeta_0$.
This lemma forms a crucial part for the proofs of theorems \ref{theorem1} and \ref{second-theorem}.
\end{remark}

\begin{remark}
The last term in the right hand side of \eqref{second-result} contains 
the contributions of {\it large} polymers of size greater than $\lambda n$ (in a sea of small polymers). It turns out that this contribution cannot be 
expanded into an absolutely convergent series,
and has to be treated non-perturbatively by counting methods.
\end{remark}

Lemma \ref{convergence-lemma} has the following consequence:

\begin{corollary}\label{lemmepol}
Suppose $r>2$. One can find $h_0>0$ independent of $n$ such that for $\vert h\vert < h_0$, 
\begin{equation}\label{coro}
\frac{1}{n}\mathbb{E}_\Gamma[\ln Z_p(\underline \eta, \widehat{\underline{\eta}})] = O(\frac{1}{n^{l(1-\kappa)-1}})
\end{equation}
\end{corollary}

\section{Convergence of the Polymer Expansion 
\eqref{eq:polymer_expansion}}

We give the main ideas of the proof of lemma \ref{convergence-lemma}.

\begin{proof}[Proof of Lemma \ref{convergence-lemma}]
A standard criterion for uniform convergence and analyticity of the polymer expansion is \cite{Brydges}
\begin{equation}
Q \equiv \sum_{t=0}^{\infty}\frac{1}{t!}\sup_{z\in
V\cup C}\sum_{\gamma\ni z, {\rm \left\vert \gamma\right\vert <\lambda
n}}\left\vert \gamma\right\vert ^{t}\zeta_0\left\vert K\left(  \gamma\right)
\right\vert < 1\,.
\end{equation}
If we prove that for polymers such that $\vert\gamma\vert <\lambda n$ we have 
\begin{equation}\label{bound-activity}
\vert K(\gamma)\vert\leq h^{\frac{c}{2}\vert\gamma\vert},
\end{equation}
then the result follows for $h$ small enough. 

The main difficulty in proving \eqref{bound-activity} is that the (optimal) estimate \eqref{actibound}, \eqref{eq: bound activity} in the Appendix shows that $K(\gamma)$ is 
not necessarily very small for graphs containing too many check nodes of maximal induced degree and too many variable 
nodes of even induced degree. More precisely for these bad graphs the activity is not exponentially small in the size of the graph. Then 
it is not possible to compensate for the ''entropy`` of the graph.

We will use an expander argument to show that these bad cases {\it do not occur} when
$\vert\gamma\vert <\lambda n$.
We derive \eqref{bound-activity}
with 
\begin{equation}\label{c}
c = r-
\frac
{2+r}{3-l(  1-\kappa)}\,.
\end{equation}
In the process of this derivation one has to require $3-l(1-\kappa) >0$
and $c>0$. This imposes the condition on the expansion constant
$\kappa >1-\frac{2(r-1)}{lr}$. Note that an expansion constant cannot be greater than 
$1-1/l$, so it is fortunate that we have  
$1-\frac{1}{l} > 1-\frac{2(r-1)}{lr}$ (for any $r>2$). 

Now we sketch the proof of \eqref{bound-activity} and \eqref{c}. Recall that 
$d_i(\gamma)$ (resp. $d_a(\gamma)$) is the induced degree of 
node $i$ (resp. $a$) in $\gamma$.  
The {\it type} of $\gamma$
is given by two vectors $\underline n=(n_s(\gamma))_{s=2}^l$ and $\underline m = (m_t(\gamma))_{t=2}^r$ defined as 
$n_{s}\left(  \gamma\right)  :=\left\vert \left\{  i\in \gamma\cap V \vert d_i(\gamma) =s\right\}  \right\vert $ and $m_{t}\left(  \gamma\right)
:=\left\vert \left\{  a\in \gamma\cap C \vert d_a(\gamma)
=t\right\}\right\vert$. In words, $n_s(\gamma)$ and $m_t(\gamma)$ count the number of 
variable and check nodes with induced degrees $s$ and $t$ in $\gamma$. 
Note that we have the constraints
\begin{equation}\label{case}
\begin{cases}
\vert \gamma\vert = \sum_{s=2}^l n_s(\gamma) + \sum_{t=2}^r m_t(\gamma)
\\ 
\sum_{s=2}^{l}sn_{s}(\gamma)=\sum_{t=2}^{r}tm_{t}(\gamma)
\end{cases}
\end{equation}
We apply the expander property to the set $\mathcal{V=}\left\{  i\in \gamma\cap V\right\}$. This reads 
\begin{equation}\label{exp}
\left\vert \partial\mathcal{V}\right\vert \geq\kappa l\left\vert \sum_{s=2}^{l}n_{s}\left(  \gamma\right)  \right\vert\,.
\end{equation}
On the other hand 
$
\left\vert \partial\mathcal{V}\right\vert\leq \sum_{t=2}^{r-1}m_{t}\left(  \gamma\right)  +\sum_{s=2}^{l}(l-s)n_{s}\left(
\gamma\right)
$.
With \eqref{exp} this yields the constraint 
\begin{equation}\label{cons}
\sum_{t=2}^{r-1}m_{t}\left(  \gamma\right)  +\sum_{s=2}^{l}(l-s)n_{s}\left(
\gamma\right) \geq\kappa
l\left\vert \sum_{s=2}^{l}n_{s}\left(  \gamma\right)  \right\vert \,.
\end{equation}
Using all constraints \eqref{case} and \eqref{cons} we can prove 
\begin{equation}
\sum_{t=2}^{r-1}\left(  r-t\right)  m_{t}\left(  \gamma\right)  \geq\left(
r-\frac{2+r}{3-l\left(  1-\kappa\right)  }\right)  \left\vert \gamma
\right\vert .
\end{equation}

Finally, keeping only the product over $t=2,\cdots, r-1$ in 
estimates \eqref{actibound} and \eqref{eq: bound activity} in the Appendix, we obtain \eqref{bound-activity}.
\end{proof}

\begin{proof}[Proof of Corollary \ref{lemmepol}]
Conditional on $\Gamma$ being an expander we have from the previous proof $0<Q<<1$. Then, polymer expansion techniques \cite{Brydges} allow to estimate the sum 
over $M$ in (\ref{eq:polymer_expansion}) term by term, which yields
\begin{equation}
\bigl\vert \frac{1}{n}\ln Z_{p}(\underline{\eta}, \widehat{\underline{\eta}})
\bigr\vert  \leq
(1-Q)^{-1} n^{-1}\sum_{z\in V\cup C}\sum_{{\gamma\ni z, {\rm \left\vert \gamma\right\vert <\lambda
n}}}\left\vert K\left(  \gamma\right)
\right\vert e^{\left\vert \gamma\right\vert } . 
\label{eq:polymer_inequality}
\end{equation}
If we take the expectation over graphs we cancel the sum over $z\in V\cup C$ and the $n^{-1}$. This allows to consider a sum 
of polymers rooted at one vertex. We compute this expectation by conditioning on the {\it first} event that $\Gamma$ is tree-like in a neighborhood of 
size $O(\ln n)$ around this vertex, and on the {\it second} complementary event. The second event has small probability
$O(n^{-(1-\beta)})$ for any $0<\beta<1$. Besides from \eqref{eq:polymer_inequality} and \eqref{bound-activity} 
it is easy to show that $n^{-1}\vert\ln Z_p\vert$ is bounded. For the first event we have that the smallest polymer is a cycle with $\vert\gamma\vert=O(\ln n)$.
This with \eqref{eq:polymer_inequality} and \eqref{bound-activity} implies that $n^{-1}\vert\ln Z_{p}(\underline{\eta}, \widehat{\underline{\eta}})\vert
\leq n^{-\beta\vert\ln \vert h\vert\vert}$. Combining all these remarks with the fact that $\Gamma$ is an expander with probability $1- O(n^{-(l(1-\kappa)-1)})$ we obtain \eqref{coro}.
\end{proof}

\section{Probability estimates on graphs\label{sec:probability on graphs}}

In this section we deal with the contribution $R(\underline\eta, \widehat{\underline\eta} )$ corresponding to terms containing at least one large polymer in 
\eqref{eq: z polymer}. We have
 \begin{equation}
  \sum_{g\subset \Gamma} K(g) = Z_p(\underline\eta, \widehat{\underline\eta}) +  R(\underline\eta, \widehat{\underline\eta}),
 \end{equation}
where 
\begin{equation}
R(\underline\eta, \widehat{\underline\eta}) 
= \sum_{g\subset\Gamma {~\rm s.t~} \exists \gamma\subset g {~\rm with~} \vert\gamma\vert\geq\lambda n} K(g),
\end{equation}

The next lemma shows that the contribution from large polymers is exponentially small, with high probability with respect 
to the graph ensemble.

\begin{lemma}
Fix $\delta>0$. Assume $l\geq 3$ odd and $l<r$. There exists a constant $C>0$ depending only on $l$ and $r$ such that for $h$ small enough%
\begin{equation}
\mathbb{P}\left[  \vert R(\underline\eta, \widehat{\underline\eta})\vert  \geq\delta\right]
\leq\frac{1}{\delta}e^{-C n} \label{eq:proba_large_loop}%
\end{equation}
\end{lemma}

\begin{proof}[Sketch of Proof]
Let $\Omega_{\Gamma}\left(  \underline{n},\underline
{m}\right)$ be the set of all $g\subset\Gamma$ with prescribed type $(\underline n(g), \underline m(g))$.
By (\ref{eq: bound activity}) and the
Markov bound
\begin{align}
\mathbb{P} & \left[  \sum_{g\subset\Gamma {~\rm with~} \vert g\vert\geq\lambda n}
\left\vert K\left(  g\right)  \right\vert \geq\delta\right]   \nonumber\\
&  \leq\frac{1}{\delta}\sum_{\vec{n},\vec{m}\in\Delta}\overline{K}\left(
\underline{n},\underline{m}\right)  \mathbb{E}_\Gamma\left[  \left\vert
\Omega_{\Gamma}\left(  \underline{n},\underline{m}\right)  \right\vert \right]  ,
\label{eq:activity_vs_entropy}%
\end{align}
Notice that the probability in \eqref{eq:activity_vs_entropy} is an upper bound on the probability in \eqref{eq:proba_large_loop}.
In \eqref{eq:activity_vs_entropy} we have 
\begin{align}
\Delta\equiv\biggl\{\left(  \underline
{n},\underline{m}\right)  \mid & \lambda n\leq\sum_{s=2}^{l}n_{s}+\sum_{t=2}^{r}m_{t},
\sum_{s=2}^{l}sn_{s}=\sum_{t=2}^{r}tm_{t}, \nonumber \\
&
\sum_{s=2}^{l}n_{s}<n,
\sum_{t=2}^{r}m_{t}<nl/r\biggr\}. 
\end{align}
The expectation of the number of $g\subset\Gamma$ with prescribed type can be estimated by combinatorial bounds provided by 
McKay \cite{McKay}. It turns out that these subgraphs proliferate exponentially in $n$ only 
for a subdomain of $\Delta$ where $\overline{K}\left(
\underline{n},\underline{m}\right)$ is exponentially smaller in $n$. In the subdomain where $\overline{K}\left(
\underline{n},\underline{m}\right)$ is not small (but it is always bounded) the number of subgraphs is subexponential when $l$ is odd and $l<r$.
As a consequence {\it for $l$ odd and $l<r$}, we are able to prove that the sum on the right hand side of \eqref{eq:activity_vs_entropy} is smaller than 
$e^{-Cn}$. Unfortunately our estimates break down for $l$ even.
\end{proof}

\section{Sketch of
Proof of Theorems \ref{theorem1} and \ref{second-theorem}}\label{sec:End of the proof}

We write
\begin{equation}
 \frac{1}{n}\ln\biggl\{\sum_{g\subset \Gamma} K(g)\biggr\} = \frac{1}{n}\ln Z_p(\underline\eta, \widehat{\underline\eta}) +  
\frac{1}{n}\ln \biggl( 1 + \frac{R(\underline\eta, \widehat{\underline\eta})}{Z_p(\underline\eta, \widehat{\underline\eta})}\biggr)\,.
\end{equation}

We first look at the second contribution coming from large polymers.
From corollary \ref{lemmepol} and the Markov bound, we have for any $\epsilon>0$,
\begin{equation}\label{strange}
 \mathbb{P}[ e^{-n\epsilon}\leq \frac{1}{Z_p(\underline\eta, \widehat{\underline\eta})}\leq e^{n\epsilon}] =1 - \frac{1}{\epsilon}O(n^{-(l(1-\kappa)-1)})
\end{equation}
Using inequalities \eqref{eq:proba_large_loop} and \eqref{strange}, and choosing $\delta=e^{-2n\epsilon}$ it is not difficult to show that 
(at this point one takes $2\epsilon < C$)
\begin{equation}
 \mathbb{P}\biggl[\biggl\vert\frac{R(\underline\eta, \widehat{\underline\eta})}{Z_p(\underline\eta, \widehat{\underline\eta})}\biggr\vert\geq 
e^{-n\epsilon}\biggr]\leq 
\frac{1}{\epsilon} O(n^{-(l(1-\kappa)-1)}) + e^{-n(C-2\epsilon)}\,.
\end{equation}
This allows to conclude that with probability $1- \frac{1}{\epsilon} O(n^{-(l(1-\kappa)-1)})$ 
\begin{equation}
 \frac{1}{n}\ln \biggl( 1 + \frac{R(\underline\eta, \widehat{\underline\eta})}{Z_p(\underline\eta, \widehat{\underline\eta})}\biggr)
= O(e^{-n\epsilon})\,.
\end{equation}
This already proves theorem \ref{second-theorem}.

It is now easy to show theorem \ref{theorem1}. 
There is a probability $O(n^{-(l(1-\kappa)-1)})$ that this last term is not small. However we can always show it is bounded by a constant independent 
of $n$. Indeed it is equal to the difference $n^{-1}\ln Z - f_{\rm Bethe}(\underline\eta, \widehat{\underline\eta}) - 
n^{-1}\ln Z_p(\underline\eta, \widehat{\underline\eta})$ where each term separately can be shown to be bounded by a constant independent of $n$. Furthermore, corollary \ref{lemmepol} tells us that the expectation of the absolute value of the first term on the r.h.s 
is $O(n^{-(l(1-\kappa)-1)})$. 
Combining these remarks allows to conclude the proof of theorem \ref{theorem1}.

\section{Appendix}\label{appendix}

We have 
\begin{equation}
K(g) = \prod_{i\in g\cap V}K_{i}\prod_{a\in g\cap C}K_{a}
\end{equation}
Quantities $K_{a},K_{i}$ are local
and can be computed only with BP messages. Let $m_i= \tanh(h_i+\sum_{a\in \partial i} \widehat{\eta}_{a\to i})$.
\begin{align}
K_{i}  =
\frac{(1-m_i)^{d_i(g) -1} + (-1)^{d_i(g)-1}(1+m_i)^{d_i(g) -1}}{2(1-m_i^2)^{d_i(g) -1}}
\end{align}
\begin{align}
& K_{a}  =\prod_{i\in\partial a\cap g}\sqrt{\frac{1-\tanh^2
\eta_{i\rightarrow a}}{1-\prod_{j\in\partial a\setminus i}
\tanh^2\eta_{j\rightarrow a}}}\prod_{i\in\partial a\cap g^c}\tanh\eta_{i\rightarrow a}
\nonumber \\ &
\times
\frac{1+\left(  -1\right)
^{d_a(g)}\prod_{i\in\partial a}
\tanh^{d_a(g) -1}\eta_{i\rightarrow a} }{1+\prod_{i\in\partial a}
\tanh\eta_{i\rightarrow a}}
\prod_{i\in\partial a\cap g} \sqrt{1-m_i^2}
\end{align}
Using these formulas and the BP equations we derive the following estimate for $\vert h_i\vert < h_0$ small enough
\begin{equation}\label{actibound}
\vert K(g)\vert \leq \overline{K}(\underline{n}(g), \underline{m}(g))
\end{equation}
where
\begin{align}
&  \overline{K}(\underline{n}(g), \underline{m}(g)) = \left(
1-\alpha_{r}rh^{2}\right)  ^{m_{r}(g)}\prod_{t=2}^{r-1}\left(  \alpha_{t}%
h^{r-t}\right)  ^{m_{t}(g)}\nonumber\\
&  \times\prod_{\substack{s=2,\\\mathrm{even}}}^{l-1}\left(  1+\frac{\beta
_{s}}{2}s\left(  s-1\right)  h^{2}\right)  ^{n_{s}(g)}\prod
_{\substack{s=3,\\\mathrm{odd}}}^{l}\left(  \beta_{s}\left(  s-1\right)
h\right)  ^{n_{s}(g)}. \label{eq: bound activity}%
\end{align}
Here $0<\alpha_{r}<1$, $\alpha_{t}>1$, $\beta_{t}>1$ are fixed numerical constants (that we can take close to $1$).
Estimate (\ref{eq: bound activity}) is essentially optimal for small $h$ as
can be checked by Taylor expanding $K(g)$ in powers of
$h_{i}$.

\vskip 0.35cm
{\bf Acknowlegment.} The work of M.V was supported by the Swiss National
Science Foundation grant no 200021-121903.

\end{document}